\documentclass[11pt]{llncs}
\usepackage{amsmath,epsfig}
\usepackage{fullpage}

\usepackage{times}
\usepackage[ruled,vlined,linesnumbered]{algorithm2e}
\usepackage{qtree}
\usepackage{subfigure}
\usepackage{url}

\usepackage[usenames,dvipsnames]{pstricks}
\usepackage{epsfig}
\usepackage{pst-grad} % For gradients
\usepackage{pst-plot} % For axes

\newtheorem{fact}{Fact}%[section]

\newcommand{\qedsymb}{\hfill{\rule{2mm}{2mm}}}

\DeclareMathOperator{\sus}{\mathit{SUS}}
\DeclareMathOperator{\lsus}{\mathit{LSUS}}
\DeclareMathOperator{\sa}{\mathit{SA}}
\DeclareMathOperator{\rank}{\mathit{Rank}}
\DeclareMathOperator{\lcp}{\mathit{LCP}}
\DeclareMathOperator{\sls}{\mathit{SLS}}
\DeclareMathOperator{\can}{\mathit{Candidate}}

\newcommand{\remove}[1]{}

\title{Shortest Unique Substring Query Revisited\thanks{All 
    missed proofs and pseudocode can be found in the appendix.}}

\author{Atalay Mert \.{I}leri$^\dag$ \and M. O\u{g}uzhan
  K\"{u}lekci$^\ddag$ \and Bojian Xu$\S$\thanks{Corresponding
    author. Supported in part by EWU Faculty Grants for
    Research and Creative Works. }}

\institute{
$^\dag$Department of Computer Engineering, Bilkent University, Turkey\\
$^\ddag$T\"UB\.ITAK National Research Institute of Electronics and Cryptology, Turkey\\
$\S$ Department of Computer Science, Eastern Washington University, USA\\
\email{aileri@bilkent.edu.tr, oguzhan.kulekci@tubitak.gov.tr, bojianxu@ewu.edu}
}

\date{}

\begin{document}
\maketitle

\begin{abstract} 
  We revisit the problem of finding shortest unique substring (SUS)
  proposed recently by~\cite{PWY-ICDE2013}.  We propose an optimal
  $O(n)$ time and space algorithm that can find an SUS for every
  location of a string of size $n$. Our algorithm significantly
  improves the $O(n^2)$ time complexity needed by~\cite{PWY-ICDE2013}.
  We also support finding all the SUSes covering every location,
  whereas the solution in~\cite{PWY-ICDE2013} can find only one SUS
  for every location.  Further, our solution is simpler and easier to
  implement and can also be more space efficient in practice, since we
  only use the inverse suffix array and longest common prefix array of
  the string, while the algorithm in~\cite{PWY-ICDE2013} uses the
  suffix tree of the string and other auxiliary data structures. Our
  theoretical results are validated by an empirical study that shows
  our algorithm is much faster and more space-saving than
  the one in~\cite{PWY-ICDE2013}.
%  Our work is
%  independently developped using different method without knowing
%  another recent $O(n)$ time solution~\cite{TIBT2014}. Experimental results
%  show our solution is competitive with this recent work.
\end{abstract}

% NOTE keywords are not used for conference papers so do not populate them
% \begin{keywords}
% keyword-1, keyword-2, keyword-3
% \end{keywords}
%

\section{Introduction}
\label{sec:intro}
Repetitive structure and regularity finding \cite{jewels} has received
much attention in stringology due to its comprehensive applications in
different fields, including natural language processing, computational
biology and bioinformatics, security, and data compression.  However,
finding the shortest unique substring (SUS) covering a given string
location was not studied, until recently it was proposed by Pei
\emph{et.\ al.}~\cite{PWY-ICDE2013}.  As pointed out
in~\cite{PWY-ICDE2013}, SUS finding has its own important usage in
search engines and bioinformatics. We refer readers
to~\cite{PWY-ICDE2013} for its detailed discussion on the applications
of SUS finding. Pei \emph{et.\ al.} proposed a solution that costs
$O(n^2)$ time and $O(n)$ space to find a SUS for every location of a
string of size $n$. In this paper, we propose an optimal $O(n)$ time
and space algorithm for SUS finding. Our method uses simpler data
structures that include the suffix array, the inverse suffix array,
and the longest common prefix array of the given string, whereas the
method in~\cite{PWY-ICDE2013} is built upon the suffix tree data
structure. Our algorithm also provides the functionality of finding
all the SUSes covering every location, whereas the method
of~\cite{PWY-ICDE2013} searches for only one SUS for every location.
Our method not only improves their results theoretically, the
empirical study also shows our method gains space saving by a factor
of 20 and a speedup by a factor of four. The speedup gained by our
method can become even more significant when the string becomes longer
due to the quadratic time cost of~\cite{PWY-ICDE2013}. Due to the very
high memory consumption of~\cite{PWY-ICDE2013}, we were not able to
run their method with massive data on our machine.

\subsection{Independence of our work.}
After we posted an initial version of this submission at
arXiv.org~\cite{IKX-CORR13} on December 11, 2013, we were contacted
via emails by the coauthors of~\cite{TIBT2014} and
\cite{HPT-submitted}, both of which solved the SUS finding using
$O(n)$ time and space. By the time we communicated, \cite{TIBT2014}
has been accepted but has not been published and \cite{HPT-submitted}
was still under review.  We were also offered with their paper drafts
and the source code of~\cite{TIBT2014}.  The methods for SUS finding
in both papers are based on the search for \emph{minimum unique
  substrings} (MUS), as what~\cite{PWY-ICDE2013} did. Our algorithm
takes a completely different approach and does not need to search for
MUS. The problem studied by \cite{HPT-submitted} is also more general,
in that they want to find SUS covering a given chunk of locations in
the string, instead of a single location considered
by~\cite{PWY-ICDE2013,TIBT2014} and our work.  So, by all means, our work is
independent and presents a different optimal algorithm for SUS
finding.  We also have included the performance comparison with the
algorithm of~\cite{TIBT2014} in the empirical study. The algorithm
from~\cite{HPT-submitted} cannot be empirically studied as the author
did not prefer to release the code until their paper is accepted.

\section{Preliminary}
\label{sec:prelim}

We consider a {\bf string} $S[1\ldots n]$, where
each character $S[i]$ is drawn from an
alphabet $\Sigma=\{1,2,\ldots, \sigma\}$. 
A {\bf substring} $S[i\ldots j]$
of $S$ represents $S[i]S[i+1]\ldots S[j]$ if $1\leq i\leq j \leq n$,
and is an empty string if $i>j$.
String $S[i'\ldots j']$ is a {\bf proper substring} of another string
$S[i\ldots j]$ if $i\leq i' \leq j' \leq j$ and $j'-i' < j-i$. 
The {\bf length} of a non-empty substring $S[i\ldots j]$, denoted as
$|S[i\ldots j]|$, is $j-i+1$. We define the length of an empty string
is zero. 
A {\bf prefix} of $S$ is a substring $S[1\ldots i]$
for some $i$, $1\leq i\leq n$. 
A {\bf proper prefix} $S[1\ldots i]$ is a prefix of $S$ where $i <
n$.
A {\bf suffix} of $S$ is a substring
$S[i\ldots n]$ for some $i$, $1\leq i\leq n$.  
A {\bf proper suffix} $S[i\ldots n]$ is a suffix of $S$ where $i >
1$.
We say the character $S[i]$ occupies the string {\bf location} $i$.
We say the substring $S[i\ldots j]$ {\bf covers} the $k$th location of
$S$, if $i\leq k \leq j$.  
For two strings $A$ and $B$, we write ${\bf A=B}$ (and say $A$ is {\bf
  equal} to $B$), if $|A|= |B|$ and $A[i]=B[i]$ for 
$i=1,2,\ldots, |A|$.  
%We write ${\bf A\subseteq B}$ if $A$ is equal to
%a substring of $B$, and ${\bf A\subset B}$ if $A$ is equal to 
%a proper substring of $B$.
%
We say $A$ is lexicographically smaller than $B$,
denoted as ${\bf A < B}$, if (1) $A$ is a proper prefix of $B$, or (2)
$A[1] < B[1]$, or (3) there exists an integer $k > 1$ such that
$A[i]=B[i]$ for all $1\leq i \leq k-1$ but $A[k] < B[k]$.
A substring
$S[i\ldots j]$ of $S$ is {\bf unique}, if there does not exist
another substring $S[i'\ldots j']$ of $S$, such that 
$S[i\ldots j] = S[i'\ldots j']$ but $i\neq i'$. 
A substring is a {\bf repeat} if it is not unique. 
%To ease presentation, we define an empty substring is a repeat. 

\begin{definition}
\label{def:sus}
For a particular string location $k\in \{1,2,\ldots, n\}$,  
the {\bf shortest unique substring (SUS) covering location} ${\bf k}$, denoted
as $\mathit{\bf SUS}_{\bf k}$, is 
a unique substring $S[i\ldots j]$, such that (1) $i\leq k \leq j$, and 
(2) there is no other unique substring $S[i'\ldots j']$ of $S$, such
that $i'\leq k \leq j'$ and $j'-i' < j-i$. 
\end{definition}
For any string location $k$, $\sus_k$ must exist, because the string
$S$ itself can be $\sus_k$ if none of the proper substrings of $S$ is
$\sus_k$. Also there might be multiple candidates for $\sus_k$. For
example, if $S={\tt abcbb}$, then $\sus_2$ can be either $S[1,2]={\tt
  ab}$ or $S[2,3]={\tt bc}$.

For a particular string location $k\in \{1,2,\ldots, n\}$, the
{\bf left-bounded shortest unique substring (LSUS) starting at location $k$},
denoted as ${\mathit{\bf LSUS}_{\bf k}}$, is a unique substring $S[k\ldots j]$,
such that either $k=j$ or any proper prefix of $S[k\ldots j]$ is not
unique. 
Note that $\lsus_1=\sus_1$, which always exists.  However, for an
arbitrary location $k\geq 2$, $\lsus_k$ may not exist. For example, if
$S={\tt abcabc}$, then none of $\{\lsus_4, \lsus_5, \lsus_6\}$
exists. 
%Example 2: if $S$ is a sequence of one identical character, then
%all $\lsus_{2\ldots n}$ do not exist.
An {\bf up-to-$j$ extension of} ${\mathit{\bf LSUS}_{\bf k}}$, denoted 
as
${\mathit{\bf LSUS}_{\bf k}^{\bf j}}$,
is the substring $S[k\ldots j]$, where  $k+|\lsus_k| \leq j \leq n$.

%Clearly, any extension of an LSUS is unique. If an
%LSUS is a suffix of $S$, it cannot be extended.

\remove{

\begin{table}
\center
\def\0{\phantom{0}}
{\footnotesize
\begin{tabular}{c|c|c|l}
\hline 
$i$ & $\lcp[i]$  & $\mathit{\sa}[i]$ & suffixes\\
\hline
\hline
%\strut&&&&\\
\01 & 0 & 11\0  &{\tt i}\\
\02 & 1 & \08\0  & {\tt  ippi}\\
\03 & 1 & \05\0  & {\tt  issippi}\\
\04 & 4 & \02\0  & {\tt  ississippi}\\
\05 & 0 & \01\0  & {\tt  mississippi}\\
\06 & 0 & 10\0  & {\tt  pi}\\
\07 & 1 &  \09\0  & {\tt ppi}\\
\08 & 0 & \07\0  & {\tt sippi}\\
\09 & 2  & \04\0  & {\tt sissippi}\\
10 & 1  & \06\0  & {\tt ssippi}\\
11 & 3 & \03\0  & {\tt ssissippi}\\
12 & 0 & -- & --\\
\hline
\end{tabular}
}
\bigskip
\caption{The suffix array and the lcp array of an example string $S={\tt mississippi}$.}
\label{tab:suflcp}
\end{table}

}

The {\bf suffix array} $\sa[1\ldots n]$ of the string $S$ is a
permutation of $\{1,2,\ldots, n\}$, such that for any $i$ and $j$,
$1\leq i < j \leq n$, we have $S[\sa[i]\ldots n] < S[\sa[j]\ldots n]$.
That is, $\sa[i]$ is the starting location of the $i$th suffix in
the sorted order of all the suffixes of $S$.
The {\bf rank array} $\rank[1\ldots n]$ is the inverse of the suffix
array. That is, $\rank[i]=j$ iff $\sa[j]=i$. 
The {\bf longest common prefix (lcp) array} $\lcp[1\ldots n+1]$ is an
array of $n+1$ integers, such that for $i=2,3,\ldots, n$, $\lcp[i]$ is
the length of the lcp of the two suffixes $S[\sa[i-1]\ldots n]$ and
$S[\sa[i]\ldots n]$. We set $\lcp[1]=\lcp[n+1]=0$.  In the literature,
the lcp array is often defined as an array of $n$ integers. We include
an extra zero at $\lcp[n+1]$ is only to simplify the description 
of out upcoming
algorithms.  
%Table~\ref{tab:suflcp} shows the suffix array and the lcp
%array of the example string {\tt mississippi}.
%
The next Lemma~\ref{lem:lsus} shows that, by using the rank array and
the lcp array of the string $S$, it is easy to calculate any $\lsus_i$ if
it exists or to detect that it does not exist.

\begin{lemma}
\label{lem:lsus}
For $i=1,2,\ldots,n$: 
$$
\lsus_i = 
\left \{
\begin{array}{ll}
S[i\ldots i + L_i], & \textrm{\ \ \ if \ \ } i+L_i \leq n\\
\emptyset & \textrm{\ \ \ otherwise}
\end{array}
\right.
$$
where $L_i = \max\{\lcp[\rank[i]],\lcp[\rank[i]+1]\}$ and $\emptyset$
means $\lsus_i$ does not exist. 
\end{lemma}

\section{SUS Finding for One Location}
\label{sec:one}

In this section, we want to find the SUS covering a given 
location $k$ using $O(k)$ time and space. We start with finding
the leftmost one if $k$ has multiple SUSes. In the end,
we will show a trivial extension to find all the SUSes
covering location $k$ with the same time
and space complexities, if $k$ has multiple SUSes.

\begin{lemma}
\label{lem:ext}
  Every SUS is either an LSUS or an extension of an LSUS.
\end{lemma}

\remove{
Example 1: $S={\tt abcbca}$, then $\sus_2 = S[1,2] =
  {\tt ab}$, which is $\lsus_1$.
  Example 2: $S={\tt abcbc}$, then $\sus_2 = S[1,2] =
  {\tt ab}$, which is an extension of $\lsus_1=S[1]$ up to
  location $2$.
}

\begin{algorithm}[t]
{\small
  \caption{Find $\sus_k$. Return the leftmost one if $k$ has multiple SUSes.}
\label{algo:one}
\KwIn{The location index $k$, and the rank array and 
      the lcp array of the string $S$} 
\KwOut{$\sus_k$. The leftmost one will be returned if $k$ has multiple
SUSes.}

\smallskip 

$start \leftarrow 1$; $length \leftarrow n$ \tcp*{The start location
  and length of the best candidate for $\sus_k$.}
\label{line:start}
%$length \leftarrow n$ \tcp*{The length of the best candidate for $\sus_k$.}
%\label{line:length}

\smallskip 

\For{$i=1, \ldots, k$\label{line:for}}{
  $L \leftarrow \max\{\lcp[\rank[i]],\lcp[\rank[i]+1]\}$\;
 \If(\tcp*[f]{$\lsus_i$ exists.}){$i+L\leq n$\label{line:if-1}}{ 
    \tcc{Extend $\lsus_i$ up to $k$ if needed. Resolve the tie by picking the leftmost SUS.}
    \If{$\max\{L+1,k-i+1\} < length$\label{line:if-2}}{
      $start \leftarrow i$;
      $length \leftarrow \max\{L+1,k-i+1\}$\;
    }
  }
  \lElse(\tcp*[f]{Early stop.}){break\label{line:earlystop}}
  
}
Print {$\sus_k \leftarrow (start,length)$}\;
}%\small
\end{algorithm}

By Lemma~\ref{lem:ext}, we know $\sus_k$ is either an LSUS or an
extension of an LSUS, and the starting location of that LSUS 
must be on or before location $k$. Then the algorithm for finding
$\sus_k$ for any given string location $k$ is simply to calculate
$\lsus_1, \lsus_2, \ldots, \lsus_k$ if existing, using Lemma~\ref{lem:lsus}. During this calculation,
 if any
LSUS does not cover the location $k$, we simply extend that LSUS up
to location $k$. We will pick the shortest one
among all the LSUS or their up-to-$k$
extensions as $\sus_k$. We resolve the tie by picking the leftmost candidate. 
It is possible this procedure can early stop if it finds an LSUS does not exist,
because that indicates all the other remaining LSUSes do not exist either.
Algorithm~\ref{algo:one} gives the pseudocode of this
procedure, where we represent $\sus_k$ by
its  two attributes: {\tt start} and {\tt length}, the starting
location and the length of $\sus_k$, respectively.

\begin{lemma}
\label{lem:one}
Given a string location  $k$ and the rank and the lcp array of the
string $S$, Algo.\ref{algo:one} can find $\sus_k$ using $O(k)$ time.
If there are multiple candidates for $\sus_k$, the leftmost one is returned. 
\end{lemma}

\begin{theorem}
\label{thm:one}
For any location $k$ in the string $S$, we can find $\sus_k$ using
$O(n)$ time and  space.  If there are multiple candidates for
$\sus_k$, the leftmost one is returned.
\end{theorem}

%\noindent
%{\bf Challenge: Can we do it $o(k)$ time with or without additional
%  data structures ?}

It is trivial to extend Algorithm~\ref{algo:one} to find all the SUSes
covering location $k$ as follows. We can first use
Algo.~\ref{algo:one} to find the leftmost $\sus_k$. Then we start
over again to recheck $\lsus_1\ldots \lsus_k$ or their up-to-$k$
extensions, and return those whose length is equal to the length of
$\sus_k$. Due to the page limit, we show the pseudocode of this procedure
in Algorithm~\ref{algo:one-all} in the
appendix.  This  procedure clearly costs an extra $O(k)$
time. Combining the results from Theorem~\ref{thm:one}, we get the
following theorem.

\begin{theorem}
\label{thm:one-all}
We can find all the SUSes covering any given location $k$ using $O(n)$ time
and space.
\end{theorem}

\section{SUS Finding for Every Location}
\label{sec:every}

In this section, we want to find $\sus_k$ for every location
$k=1,2,\ldots,n$.  If $k$ has multiple SUSes, the leftmost one will be
returned.  In the end, we will show a trivial extension to return
all SUSes for every location.

A natural solution is to iteratively use Algo.\ref{algo:one} as a
subroutine to find every $\sus_k$, for $k=1,2,\ldots,n$. However, the
total time cost of this solution will be $O(n)+\sum_{k=1}^n O(k) =
O(n^2)$, where $O(n)$ captures the time cost for the construction of
the rank array and the lcp array and $\sum_{k=1}^n O(k)$ is the total
time cost for the $n$ instances of Algo.~\ref{algo:one}. We want to
have a solution that costs a total of $O(n)$ time and space, which follows
that the amortized cost for finding each SUS is $O(1)$.

By Lemma~\ref{lem:ext}, we know that every SUS must be an LSUS or an
extension of an LSUS.  The next Lemma~\ref{lem:ext2} further says if
$\sus_k$ is an extension of an LSUS, it has special properties and can
be quickly obtained from $\sus_{k-1}$.

\begin{lemma}
\label{lem:ext2}
For any $k\in \{2,3,\ldots,n\}$, if $\sus_k$ is an extension of an
LSUS, then (1) $\sus_{k-1}$ must be a substring whose right boundary
is the character $S[k-1]$, and (2) $\sus_k$ is the substring $\sus_{k-1}$
appended by the character $S[k]$.
\end{lemma}

\subsection{The overall strategy}
\label{subsec:strategy}
We are ready to 
present the overall strategy for finding SUS of every location, 
by using  Lemma~\ref{lem:ext} and~\ref{lem:ext2}.
We will calculate all the SUS in the order of $\sus_1, \sus_2, \ldots,
\sus_n$.  That means when we want to calculate $\sus_k$, $k\geq 2$, we
have had $\sus_{k-1}$ calculated already.  Note that $\sus_1 =
\lsus_1$, which is easy to calculate using Lemma~\ref{lem:lsus}.  Now
let's look at the calculation of a particular $\sus_k$, $k\geq 2$.  By
Lemma~\ref{lem:ext}, we know $\sus_k$ is either an LSUS or an
extension of an LSUS.  By Lemma~\ref{lem:ext2}, we also know if
$\sus_k$ is an extension of an LSUS, then the right boundary of
$\sus_{k-1}$ must $S[k-1]$ and $\sus_k$ is just $\sus_{k-1}$ appended
by the character $S[k]$. Suppose when we want to calculate $\sus_k$, we
have already calculated the shortest LSUS covering location $k$ or
have known the fact that no LSUS covers location $k$.  Then, by using
$\sus_{k-1}$, which has been calculated by then, and the shortest LSUS
covering location $k$, we will be able to calculate $\sus_k$ as
follows:

Case 1: If the right boundary of $\sus_{k-1}$ is not $S[k-1]$, then we
know $\sus_k$ cannot be an extension of an LSUS (the contrapositive of
Lemma~\ref{lem:ext2}). Thus, $\sus_k$ is just the shortest LSUS
covering location $k$, which must be existing in this case.

Case 2: If the right boundary of $\sus_{k-1}$ is $S[k-1]$, then
  $\sus_k$ may or may not be an extension of an LSUS.  We will
  consider two possibilities:
(1) If the shortest LSUS covering location $k$ exists, we will
  compare its length with $|\sus_{k-1}|+1$, and pick the shorter one
  as $\sus_k$. If both have the same length, we solve the tie by
  picking the one whose starting location index is smaller.
(2)  If no LSUS covers location $k$, $\sus_k$
  will just be $\sus_{k-1}$ appended by $S[k]$.

  Therefore, the real challenge here, by the time we want to calculate
  $\sus_k$, $k\geq 2$, is to ensure that we would already have
  calculated the shortest LSUS covering location $k$ or we would
  already have known the fact that no LSUS covers location $k$.

\subsection{Preparation}
\label{subsec:lsus}
We now focus on the calculation of the shortest LSUS covering every
string location k, denoted by ${\mathit{\bf SLS_k}}$.
Let $\mathit{\bf Candidate_i^k}$ denote the shortest one among
those of $\{\lsus_1, \ldots, \lsus_k\}$ that exist and cover location $i$.   
The leftmost one will be picked if multiple choices exist for both 
$\sls_k$ and $\can_i^k$.
For an arbitrary $k$, $1\leq k \leq n$, $\sls_k$ may not exist,
because the location $k$ may not be covered by any LSUS. However, 
if $\sls_k$ exists, by the definition of  
$\sls$ and $\can$, we have:

\begin{fact}
\label{fact:can}
$\sls_k=\can_k^{k}=\can_k^{k+1}=\cdots =\can_k^{n}$, if $\sls_k$ exists.
\end{fact}

Our goal is to ensure $\sls_k$ will have been known when we want to
calculate $\sus_k$, so we calculate every $\sls_k$ following the same
order $k=1,2,\ldots,n$, at which we calculate all SUSes.
Because we need to know every $\lsus_i$, $i\leq k$ in order to
calculate $\sls_k$ (Fact~\ref{fact:can}), we will walk through the
string locations $k=1,2,\ldots,n$: at each walk step $k$, we calculate
$\lsus_k$ and maintain $\can_i^k$ for every string location $i$ that has been covered by 
at least one of $\lsus_1,\lsus_2,\ldots, \lsus_k$. 
Note that $\can_i^k = \sls_i$ for every $i\leq k$
(Fact~\ref{fact:can}). Those $\can_i^k$ with $i \leq k$ would have been
used as $\sls_i$ in the calculation of $\sus_i$. So, after each
walk step $k$, we will only need to maintain the candidates  for
location  after $k$.

\begin{lemma}
\label{lem:exist}
(1) $\lsus_1$ always exists. (2) If $\lsus_k$
exists, then $\{\lsus_1, \lsus_2, \ldots, \lsus_{k}\}$ all exist. 
(3) If $\lsus_k$
does not exist, then none of $\{\lsus_k, \lsus_{k+1}, \ldots, \lsus_n\}$ exists.   
\end{lemma}

%Let's look at the situation at the end of a particular walk step $k$.
We know at the end of the $k$th walk step, we have calculated $\lsus_1, \lsus_2,
\ldots, \lsus_k$. By Lemma~\ref{lem:exist}, we know that there exists some
$\ell_k$, $1\leq \ell_k\leq k$, such that
$\lsus_1,\ldots,\lsus_{\ell_k}$ all exist, but $\lsus_{\ell_k+1}\ldots
\lsus_k$ do not exist. If $\ell_k=k$, that means $\lsus_1,\ldots,\lsus_k$
all exist. Let $\gamma_k$ denote the right boundary of
$\lsus_{\ell_k}$, i.e., $\lsus_{\ell_k} = S[\ell_k \ldots \gamma_k]$. We
know every location $j= 1, \ldots, \gamma_k$ has its candidate $\can_j^k$
calculated already, because every such location $j$ has been covered by at least
one of the LSUSes among $\lsus_1,\ldots,\lsus_{\ell_k}$. We also know
if $\gamma_k < n$, every location $j=\gamma_k+1,\ldots,n$ still does
not have its candidate calculated, because every such location $j$ has not been
covered by any LSUS from $\lsus_1,\ldots,\lsus_{\ell_k}$ that we have
calculated at the end of the $k$th walk step. 

\begin{lemma}
\label{lem:cover}
At the end of the $k$th walk step, 
if $\gamma_k > k$:
 for any $i$ and $j$, $k \leq i < j \leq
\gamma_k$, $\can_j^k$ also covers location $i$.
\end{lemma}

\begin{proof}
$\can_j^k$ is a substring starting somewhere on or before $k$ and
going through the location $j$. Because $k\leq i < j$, it is obvious
that $\can_j^k$ goes through location $i$. \qed
\end{proof}

\begin{lemma}
\label{lem:canlength}
At the end of the $k$th walk step, 
if $\gamma_k > k$, then  
$|\can_k^k| \leq |\can_{k+1}^k| \leq \ldots \leq |\can_{\gamma_k}^k|$.
\end{lemma}

  \begin{proof}
    By Lemma~\ref{lem:cover}, we know $\can_j^k$ also covers location
    $i$, for any $i$ and $j$, $k\leq i < j \leq \gamma_k$.  Thus, if
    $|\can_j^k| < |\can_i^k|$, location $i$'s current candidate should
    be replaced by location $j$'s candidate, because that gives
    location $i$ a shorter (better) candidate. However, the current
    candidate for location $i$ is already the shortest candidate. It
    is a contradiction.  So, $|\can_i^k| \leq |\can_j^k|$, which proves
    the lemma. \qed
  \end{proof}

\remove{

\begin{lemma}
\label{lem:canshare}
If $\gamma_k \geq k$: for any $i$ and $j$, $k \leq i < j \leq
\gamma_k$, if $|\can_i^k| = |\can_j^k|$ but $\can_i^k \neq \can_j^k$, 
we can replace location $i$'s candidate with location $j$'s candidate.
\end{lemma}

\begin{proof}
  By Lemma~\ref{lem:cover}, we know $\can_j^k$ also covers location
  $i$.  Thus, if $|\can_i^k| = |\can_j^k|$, we can replace location
  $i$'s candidate with $\can_j^k$.
\end{proof}

}

\begin{lemma}
\label{lem:lsus2}
For each $i = 2,3,\ldots,n$: $ |\lsus_i| \geq |\lsus_{i-1}|-1 $
\end{lemma}

\begin{proof}
  We prove the lemma by contradiction. Suppose $\lsus_{i-1} = S[i-1 \ldots
  j]$ for some $j$, $i-1\leq j\leq n$.  If $|\lsus_i| <
  |\lsus_{i-1}|-1$, it means $\lsus_i = S[i \ldots k]$, where $i\leq
  k<j$. Because $S[i \ldots k]$ is unique, $S[i-1 \ldots k]$ is also unique, whose
  length however is shorter than $S[i-1 \ldots j]$.  This is a contradiction
  because $S[i-1 \ldots j]$ is already $\lsus_{i-1}$. Thus, the claim in the
  lemma is true. \qed
\end{proof}

\subsection{Finding $\sls$ for every location}

\begin{algorithm}[t]
{\footnotesize
  \caption{The sequence of function calls 
$\mathit{FindSLS(1)}, \mathit{FindSLS(2)}, \ldots,
\mathit{FindSLS(n)}$
returns $\sls_1, \sls_2, \ldots, \sls_n$, if the corresponding $\sls$ exists;
otherwise, {\tt null} will be returned.}
\label{algo:list}

%\tcc{Static variables}

\smallskip 

Construct $Rank[1\ldots n]$ and $LCP[1\ldots n]$ of the string $S$\;
Initialize an empty $List$\tcp*{Each node has four fields: \{{\tt ChunkStart, ChunkEnd, start, length}\}.}
$head \leftarrow  0$; $tail \leftarrow  0$ \tcp*{Reference to the head and tail node of the $List$}

\bigskip 

\SetKwIF{If}{ElseIf}{Sls}{if}{then}{else if}{$\mathit{FindSLS}(k)$}{endif}

\Sls{
\tcc{Process $\lsus_k$, if it exists.}
\smallskip 

$L \leftarrow \max\{\lcp[\rank[k]], \lcp[\rank[k]+1]\}$\;

\smallskip 

\If(\tcp*[f]{$\lsus_k$ exists.}){$k+L \leq n$}{
  \tcp{Add a new list element at the tail, if necessary.}
  \lIf(\tcp*[f]{$List$ was empty.}){$head=0$}{
    $List[1]\leftarrow (k,k+L,k,L+1)$; 
    $head\leftarrow 1$;
    $tail\leftarrow 1$
  }
  \ElseIf{$k + L > List[tail].ChunkEnd$}{
    $tail ++$;
    $List[tail] \leftarrow (List[tail-1].ChunkEnd+1, k+L, k,L+1)$;
  }
  \smallskip
  \tcc{Update candidates and merge the nodes whose candidates can
    be shorter. Resolve the tie by picking the leftmost one.}
  $j \leftarrow tail$\;
  \lWhile{$j \geq head$ \textrm{and} $List[j].length >
    L+1$ \label{line:while}}
    {$j--$\;}
  $List[j+1] \leftarrow (List[j+1].ChunkStart, List[tail].ChunkEnd, k, L+1)$;
  $tail \leftarrow j+1$\; 
}  

\lIf(\tcp*[f]{The list is not empty.}){$head \neq 0$}{
  $\sls_k \leftarrow (head.start, head.length)$ 
}

\lElse{
  $\sls_k \leftarrow {\tt (null, null)}$ \tcp*{$\sls_k$ does not exist.}
}

\bigskip 

\tcc{Discard the information about location $k$ from the $List$.}

\smallskip 

\If(\tcp*[f]{$List$ is not empty}){$head > 0$\label{line:del-start}}{
  \If{$List[head].ChunkEnd \leq k$}{
    $head ++$\tcp*{Delete the current head node}
    \lIf(\tcp*[f]{$List$ becomes empty}){$head > tail$}{
      $head \leftarrow 0$; $tail \leftarrow 0$;
    }
  }
  \lElse{$List[head].ChunkStart \leftarrow k+1$\label{line:del-end}\;}
}

\Return{$\sls_k$}

}%\small
}

\end{algorithm}

\noindent
{\bf Invariant.} We calculate $\sls_k$ for $k=1,2,\ldots, n$ 
by maintaining the following invariant at the end of every walk step
$k$:
(A) If $\gamma_k > k$, locations $\{k+1,k+2,\ldots,
\gamma_k\}$ will be cut into chunks, such that: (A.1)
All locations in one chunk have the same candidate. 
%This requirement is valid because of
%Lemma~\ref{lem:canlength}. 
%(A.2) The candidate of every location in
%one chunk will be the chunk's rightmost location's
%candidate. This requirement is valid because of
%Lemma~\ref{lem:canshare}. 
(A.2) Each chunk will be represented by a
linked list node of four fields: {\tt \{ChunkStart, ChunkEnd, start,
  length\}}, respectively representing the start and end location of the
chunk and the start and length of the candidate shared by all 
locations of the chunk.  (A.3) All nodes representing different
chunks will be connected into a linked list, which has a {\tt head} 
and a {\tt tail}, referring to the two nodes that represent the lowest
positioned chunk and the highest positioned chunk. 
(B) If $\gamma_k \leq k$, the linked list is empty.  

\bigskip

\noindent
{\bf Maintenance of the invariant.} We describe in an inductive
manner the procedure that maintains the invariant. 
%During the
%procedure every $\sls_k$ will be calculated and returned.
Algo.~\ref{algo:list} shows the pseudocode of the procedure.  We
start with an empty linked list.

\bigskip 

\noindent
\emph{Base step: $k=1$.} We are making the first walk step. 
We first calculate $\lsus_1$ using
Lemma~\ref{lem:lsus}. We know $\lsus_1$ must exist.  Let's
  say $\lsus_1 = S[1\ldots \gamma_1]$ for some $\gamma_1\leq
  n$. Then, $\can_i^1 = \lsus_1$ for every $i=1,2,\ldots, \gamma_1$.
  We record all these candidates by using a single node
  $(1,\gamma_1,1,\gamma_1)$. This is the only node in the
  linked list and is pointed by both {\tt head} and {\tt tail}. We
  know $\sls_1 = \can_1^1$ (Fact~\ref{fact:can}), so we return
  $\sls_1$ by returning ${\tt (head.start,head.length)=(1,\gamma_1)}$.
  We then change {\tt head.ChunkStart} from $1$ to be $2$. If it turns
  out ${\tt head.ChunkEnd=\gamma_1 < 2}$, meaning $\lsus_1$ really
  covers location $1$ only, we delete the {\tt head} node from the
  linked list, which will then become empty.

\bigskip 

\noindent
\emph{Inductive step: $k\geq 2$.} We are making the $k$th walk step. 
We first calculate $\lsus_k$.

Case 1: $\lsus_k$ does not exist. (1) If {\tt head} does not exist.
It means location $k$ is covered neither by any of
$\lsus_1,\ldots,\lsus_{k-1}$ nor by $\lsus_k$, so $\sls_k$ simply does
not exist, and we will simply return ${\tt (null,null)}$ to indicate
that $\sls_k$ does not exist.  (2) If {\tt head} exists, we will
return ${\tt (head.start,head.length)}$ as $\sls_{k}$, because
$\can_k^k = \sls_k$ (Fact~\ref{fact:can}). Then we will remove the
information about location $k$ from the head by setting
$head.ChunkStart = k+1$. After that, we will remove the {\tt head}
node if it turns out that ${\tt head.ChunkEnd < head.ChunkStart}$.

Case 2: $\lsus_k$ exists. Let's say $\lsus_k = S[k\ldots \gamma_k]$,
  $\gamma_k \leq n$. By Lemma~\ref{lem:exist}, we know
  $\lsus_1,\ldots, \lsus_{k-1}$ all exist. Let $\gamma_{k-1}$ denote
  the right boundary of $\lsus_1,\ldots, \lsus_{k-1}$. By
  Lemma~\ref{lem:lsus2}, we know $\gamma_{k-1}$ is also the right
  boundary of $\lsus_{k-1}$, i.e., $\lsus_{k-1} =
  S[k-1 \ldots \gamma_{k-1}]$. Note that both $\gamma_{k-1} < k$ and
  $\gamma_{k-1} \geq k$ are possible.  
  (1) If {\tt head} does not exist, it means $\gamma_{k-1} < k$ and
  none of locations $\{k\ldots \gamma_k\}$ is covered by any of
  $\lsus_1,\ldots, \lsus_{k-1}$. We will insert a new node ${\tt (k,
    \gamma_k, k, \gamma_k-k+1)}$, which will be the only node in the
  linked list.
  (2) If {\tt head} exists, it means $\gamma_{k-1} \geq k$.  If ${\tt
    \gamma_k > tail.ChunkEnd = \gamma_{k-1}}$, we will first insert at
  the tail side of the linked list a new node ${\tt (tail.ChunkEnd+1,
    \gamma_k, k, \gamma_k-k+1)}$ to record the candidate information
  for locations in the chunk after $\gamma_{k-1}$ through $\gamma_k$.
  After the work in either (1) or (2) is finished, we will then travel
  through the nodes in the linked list from the tail side toward the
  head. We will stop when we meet a node whose candidate is shorter
  than or equal to $\lsus_k$ or when we reach the head end of the
  linked list. This travel is valid
  because of Lemma~\ref{lem:canlength}. We will merge all the nodes
  whose candidates are longer than $\lsus_k$ into one node.
  The chunk covered by the new node is the union of the chunks covered
  by the merged nodes, and the candidate of the new node obtained from
  merging is
  $\lsus_k$. This merge process ensures every location maintains its
  best (shortest) candidate by the end of each walk step, and also
  resolves the tie
  of multiple candidates by picking the
  leftmost one.  
  We will return ${\tt (head.start,head.length)}$ as
  $\sls_k$, because $\can_k^k = \sls_k$
  (Fact~\ref{fact:can}). Finally, we will remove the information about
  location $k$ from the head by setting $head.ChunkStart = k+1$. We
  will remove the {\tt head} node if it turns out that ${\tt
    head.ChunkEnd > head.ChunkStart}$.

\begin{lemma}%[Amortized time cost of $\mathit{FindSLS}$]
\label{lem:list-time}
Given the lcp array and the rank array of $S$, 
the amortized time cost of $\mathit{FindSLS}()$ is $O(1)$.
\end{lemma}

\begin{algorithm}[t]
{\footnotesize
  \caption{Find $\sus_k$, $k=1,\ldots, n$. The leftmost one
is returned if $k$ has multiple SUSes.}
\label{algo:every}

\smallskip 
\For{$k \leftarrow 1 \ldots n$\label{line:for1}}{
   $(start, length)\leftarrow \mathit{FindSLS}(k)$\label{line:update}
   \tcp*{$\sls_k$;  It is ${\tt (null, null)}$ if $\sls_k$ does not exist. }
    
   \smallskip 

    \lIf{$k=1$\label{line:if-1.1}}{
      Print $\sus_k \leftarrow (start, length)$\label{line:if-1.2}\;}
%      \tcp{The right boundary of $\sus_{k-1}$ is not
%        $S[k-1]$, so $\sus_k$ cannot be an LSUS  extension.}
    \lElseIf{$\sus_{k-1}.start + \sus_{k-1}.length -1 > k - 1$\label{line:if-2.1}}{
     Print $\sus_k \leftarrow (start, length) \label{line:if-2.2}$\; 
    }
    \lElseIf{$(start,length)={\tt (null,null)}$\label{line:if-3.1}}{
      Print $\sus_k \leftarrow (\sus_{k-1}.start, \sus_{k-1}.length + 1)$\label{line:if-3.2}\;
    }
    \lElseIf{$length < \sus_{k-1}.length+1$\label{line:if-5.1}}{
      Print $\sus_k \leftarrow (start, length)$\label{line:if-5.2}\;
    }
    \Else(\tcp*[f]{Resolve the tie by picking the leftmost one.})
%If{$\sus_{k-1}.length+1 < length $\label{line:if-4.1}}
    {
     Print $\sus_k \leftarrow (\sus_{k-1}.start, \sus_{k-1}.length + 1)$\label{line:if-4.2}
   }
}
}%\small   
\end{algorithm}

\subsection{Finding $\sus$ for every location}

Once we are able to sequentially calculate every
$\sls_k$ or detect  it does not exist,
we are ready to
calculate every $\sus_k$ by using the strategy described in
Section~\ref{subsec:strategy}. Algorithm~\ref{algo:every} gives the
pseudocode of the procedure. It calculates SUSes in the order
of $\sus_1,\sus_2,\ldots,\sus_n$ (Line~\ref{line:for1}). For each
location $k$, the function call at Line~\ref{line:update} is to
calculate $\sls_k$ or to find $\sls_k$ does not exist. 
Line~\ref{line:if-1.1}%--\ref{line:if-1.2} 
handles the special case
where $\sus_1=\lsus_1=\sls_1$. 
The condition at Line~\ref{line:if-2.1} shows that $\sus_i$ cannot be
an extension of an LSUS (Lemma~\ref{lem:ext2}), so $\sus_k=\sls_k$,
which must exist.
Line~\ref{line:if-3.1} %--\ref{line:if-3.2} 
 handles the case where
$\sls_k$ does not exist, so $\sus_k$ must be $\sus_{k-1}$ appended by
$S[k]$.
Line~\ref{line:if-5.1} %-\ref{line:if-5.2} 
 handles the case where $\sls_k$ is shorter than
the one-character extension of $\sus_{k-1}$, so $\sus_k$ is
$\sls_k$. 
Line~\ref{line:if-4.2} handles the case where
$\sls_k$ is longer than the one-character extension of $\sus_{k-1}$, so 
$\sus_k$ is $\sus_{k-1}$ appended by
$S[k]$.
This also revolves the tie by picking the leftmost one if $k$ is
covered by multiple SUSes.

\begin{theorem}
\label{thm:time}
Algo.~\ref{algo:every} can find $\sus_1, \sus_2, \ldots, \sus_n$ of
string $S$ using a total of $O(n)$ time and space.
\end{theorem}

\subsection{Extension: finding all the SUSes for every location.}
It is possible that a particular location can have multiple SUSes.
For example, if $S={\tt abcbb}$, then $\sus_2$ can be either
$S[1,2]={\tt ab}$ or $S[2,3]={\tt bc}$. Algorithm~\ref{algo:every}
only returns one of them and resolve the tie by picking the leftmost
one. However, it is trivial to modify
Algorithm~\ref{algo:every} to return all the SUSes of
every location, without changing Algorithm~\ref{algo:list}.

Suppose a particular location $k$ has multiple SUSes. We know, at the
end of the $k$th walk step but for the linked list update, $\sls_k$
returned by Algorithm~\ref{algo:list} is recorded by the {\tt head}
node and is the leftmost one among all the SUSes that are LSUS and
cover location $k$.  Because every string location maintains its
shortest candidate and due to Lemma~\ref{lem:canlength}, all the other
SUSes that are LSUS and cover location $k$ are being recorded by other
linked list nodes that are immediately following the {\tt head} node.
This is because if those other SUSes are not being recorded, that
means the location right after the head node's chunk has a candidate
longer than $\sus_k$ or does not have a candidate calculated yet, but
that location is indeed covered by a $\sus_k$ at the end of the $k$th
walk step. It's a contradiction. Same argument can be made to the
other next neighboring locations that are covered by $\sus_k$.

Therefore, finding all the SUSes covering location $k$ becomes
easy---simply go through the linked list nodes from the {\tt head}
node toward the {\tt tail} node and report all the LSUSes whose
lengths are equal to the length of $\sus_k$ that we have found. If
the rightmost character of $\sus_{k-1}$ is $S[k-1]$ and the substring $\sus_{k-1}$
appended by $S[k]$ has the same length, that substring 
will be reported too. 
Due to the page limit, the updated code is given in
Algorithm~\ref{algo:all} in the appendix, where the {\tt flag} is used
to note in what cases it is possible to have multiple SUSes and thus
we need to check the linked list nodes
(Line~\ref{line:all-1}--\ref{line:all-2}).  The overall time cost of
maintaining the linked list data structure (the sequence of function
calls $\mathit{FindSLS(1)}, \mathit{FindSLS(2)}, \ldots,
\mathit{FindSLS(n)}$) is still $O(n)$.  The time cost of reporting the
SUSes covering a particular location becomes $O(occ)$, where $occ$ is
the number of SUSes that cover that location.

%Experiment

\begin{figure}[t]
\begin{center}
\begin{tabular}{cc}
\includegraphics[scale=.8]{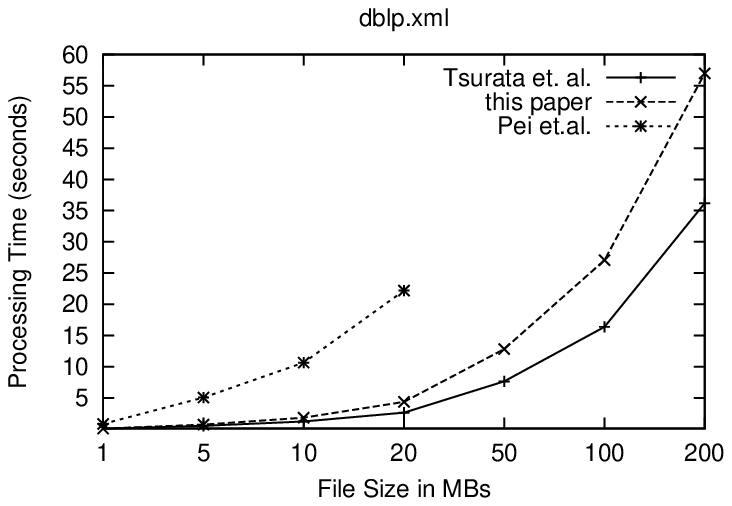} &
\includegraphics[scale=.8]{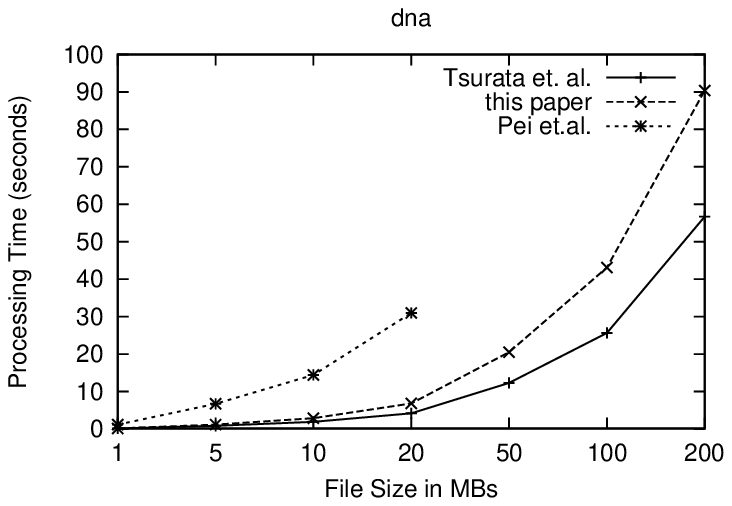} \\ 
\includegraphics[scale=.8]{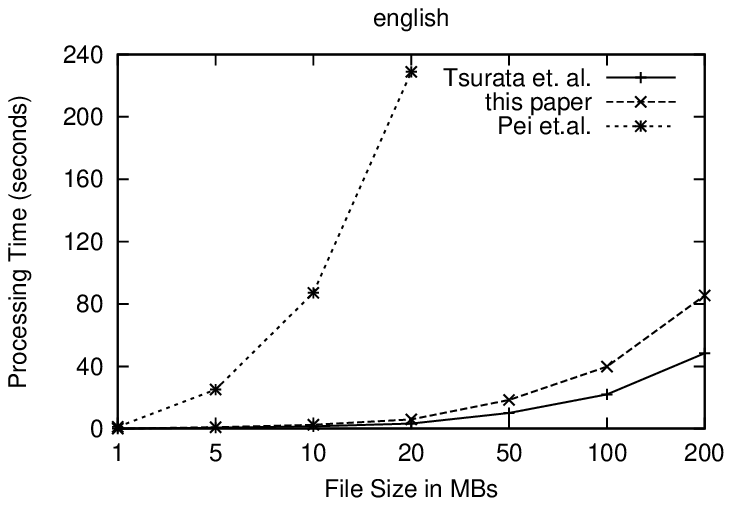} & 
\includegraphics[scale=.8]{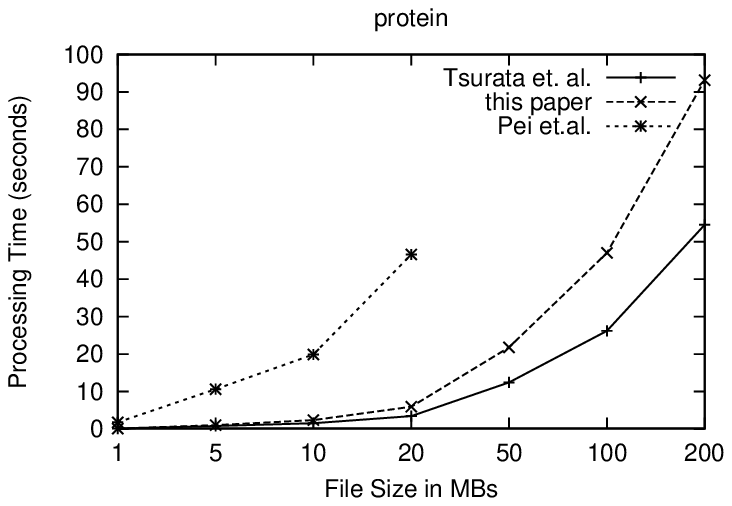} 
\end{tabular}
\end{center}
\caption{The processing speed of RSUS, OSUS, and this study's proposal
  on several files of different sizes.}
\label{fig::exp-time}
\end{figure}

\begin{figure}[t]
\begin{center}
\begin{tabular}{cc}
 \includegraphics[scale=.8]{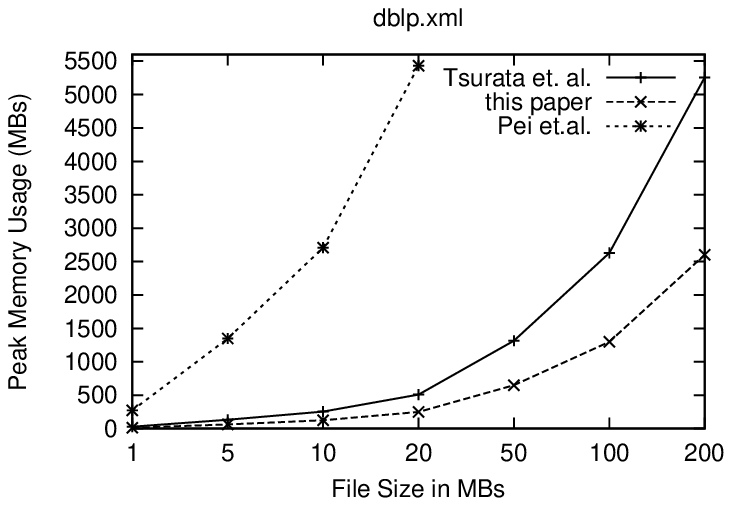} & 
\includegraphics[scale=.8]{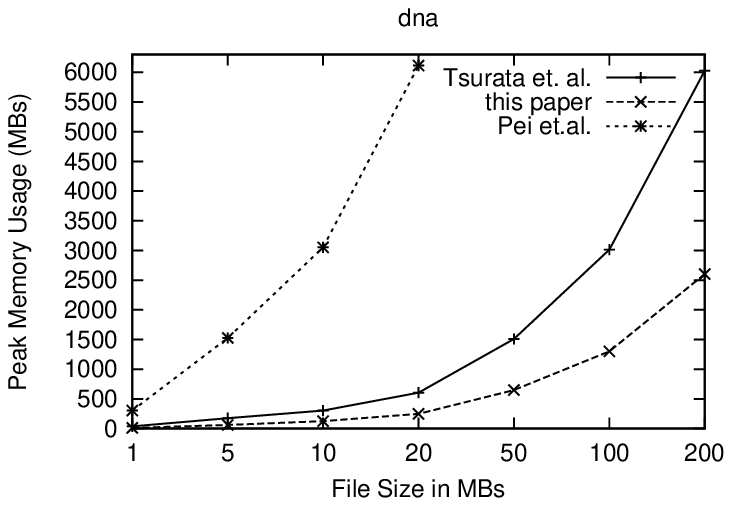} \\ 
\includegraphics[scale=.8]{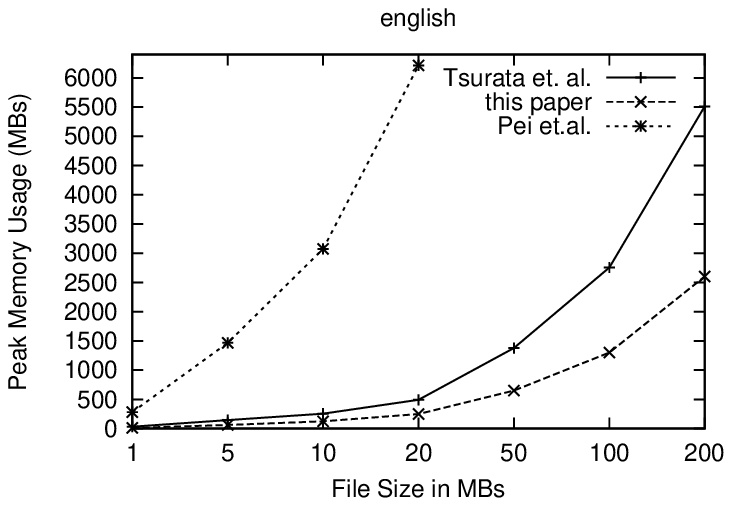} & 
\includegraphics[scale=.8]{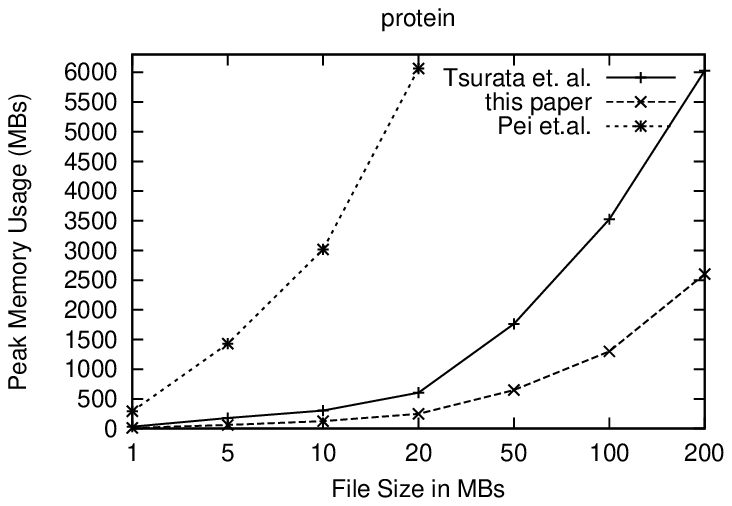} 
\end{tabular}
\end{center}
\caption{The peak memory consumptions
  of RSUS, OSUS, and this study's proposal on several files of different sizes.}
\label{fig::exp-space}
\end{figure} 

\section{Experiments}
\label{sec:exp}

We have implemented our proposal without engineering optimization effort in
C++ by using the \texttt{libdivsufsort}\footnote{Available at:
  \url{https://code.google.com/p/libdivsufsort}.} library for the
suffix array construction and Kasai \emph{et.\ al.}'s
method~\cite{KLAAP01} to compute the LCP array.  We have compared our
work against Pei \emph{et.\ al.}'s RSUS~\cite{PWY-ICDE2013} and
Tsurata \emph{et.\ al.}'s~\cite{TIBT2014} OSUS implementations, a
recent independent work obtained via personal communication. Notice
that OSUS also computes the suffix array with the same
\texttt{libdivsufsort} package.

RSUS was prepared with an R interface.  We stripped off that R
interface and build a standalone C++ executable for the sake of fair
benchmarking.  OSUS was originally developed in C++. We run it with
the \texttt{-l} option to compute a single leftmost SUS for a given
position rather than its default configuration of reporting all SUSs.
We also commented the sections that print the results to the screen on
all three programs to be able to measure the algorithmic performance
better.

We run the tests on a machine that has Intel(R) Core(TM) i7-3770 CPU @
3.40GHz processor with 8192 KB cache size and 16GB memory.  The
operating system was Linux Mint 14. We used the Pizza\&Chili
corpus%~\footnote{\url{http://pizzachili.dcc.uchile.cl/texts.html}} 
in
the experiments by taking the first 1, 5, 10, 20, 50, 100, and 200 MBs
of the largest \textit{dblp.xml, dna, english,} and \textit{protein}
files.  The results are shown in Figure~\ref{fig::exp-time} and
\ref{fig::exp-space}.

It was not possible to run the RSUS on large files, since RSUS
requires more memory than that our machine has, and thus, only up to 20MB
files were included in the RSUS benchmark.  Compared to RSUS, we have
observed that our proposal is more than 4 times faster and uses 20 times
less memory.  The experimental results revealed that OSUS is on the
average 1.6 times faster than our work, but in contrast, uses 2.6
times more memory.

The asymptotic time and space complexities of both ours and OSUS are
same as being linear (note that the $x$ axis in both
figures uses log scale).  The peak memory
usage of OSUS and ours are different although they both use
suffix array, rank array (inverse suffix array), and the LCP array,
and computing these arrays are done with the same library
(libdivsufsort).  The difference stems from different ways these
studies follow to compute the SUS.  OSUS computes the SUS by using an
additional array, which is named as the meaningful minimal unique
substring array in the corresponding study.  Thus, the space used for
that additional data structure makes OSUS require more memory.

Both OSUS and our scheme presents stable running times on all dblp,
dna, protein, and english texts and scale well on increasing sizes of
the target data conforming to their linear time complexity.  On the
other hand RSUS exhibits its $O(n^2)$ time complexity on all texts,
and especially its running time on english text takes much longer when
compared to other text types.

\section{Conclusion}
\label{sec:con}
We revisited the shortest unique substring finding problem and
proposed an optimal linear-time and linear-space algorithm for finding
the SUS for every string location. Our algorithm significantly
improved the recent work~\cite{PWY-ICDE2013} both theoretically and
empirically. Our work is independently discovered without knowing
another recent linear-time and linear-space solution that is to appear
in~\cite{TIBT2014} and uses a different approach with
competitive performance.

%\section{Acknowledgement}
%We acknowledge the author of~\cite{PWY-ICDE2013,TIBT2014} for providing their
%source code.

\vspace*{-2mm}

\bibliographystyle{splncs03}

\bibliography{bibjsv,repeat,pm}

\newpage

\section*{Appendix}

\begin{algorithm}[h!]
{\small
  \caption{Find all the SUSes covering a given location $k$.}
\label{algo:one-all}
\KwIn{The location index $k$, and the rank array and 
      the lcp array of the string $S$} 
\KwOut{All the SUSes covering location $k$.}

\smallskip 

$start \leftarrow 1$; $length \leftarrow n$ \tcp*{The start location
  and length of the best candidate for $\sus_k$.}
\label{line:start-all}
%$length \leftarrow n$ \tcp*{The length of the best candidate for $\sus_k$.}
%\label{line:length-all}

\smallskip 

\tcc{Find the length of $\sus_k$.}
\For{$i=1, \ldots, k$\label{line:for-all}}{
  $L \leftarrow \max\{\lcp[\rank[i]],\lcp[\rank[i]+1]\}$\;
 \If(\tcp*[f]{$\lsus_i$ exists.}){$i+L\leq n$\label{line:if-1-all}}{ 
    \If(\tcp*[f]{Extend $\lsus_i$ to location $k$ if necessary.})
       {$\max\{L+1,k-i+1\} < length$\label{line:if-2-all}}{
      $start \leftarrow i$;
      $length \leftarrow \max\{L+1,k-i+1\}$\;
    }
  }
  \lElse(\tcp*[f]{Early stop.}){break\label{line:earlystop-all}}
}

\tcc{Find all SUSes covering location $k$.}
\For{$i=1, \ldots, k$\label{line:for-all-2}}{
  $L \leftarrow \max\{\lcp[\rank[i]],\lcp[\rank[i]+1]\}$\;
 \If(\tcp*[f]{$\lsus_i$ exists.}){$i+L\leq n$\label{line:if-1-all-2}}{ 
    \If(\tcp*[f]{Extend $\lsus_i$ to location $k$ if necessary.})
       {$\max\{L+1,k-i+1\} = length$\label{line:if-2-all-2}}{
         Print $(i, \max\{L+1,k-i+1\})$\;
    }
  }
  \lElse(\tcp*[f]{Early stop.}){break\label{line:earlystop-all-2}}
}

}%\small
\end{algorithm}

%--------------------------------------------------

\begin{algorithm}[h!]
{\footnotesize
  \caption{Find $\sus_k$, for $k=1,\ldots, n$. All SUSes are
    returned if $k$ has multiple SUSes.}
\label{algo:all}

\smallskip 

\For{$k \leftarrow 1 \ldots n$\label{line-all:for1}}{
  $flag \leftarrow 0$;
   $(start, length)\leftarrow \mathit{FindSLS}(k)$\label{line-all:update}
   \tcp*{$\sls_k$;  ${\tt (null, null)}$ if $\sls_k$ does not exist. }
    
   \smallskip 

    \lIf{$k=1$\label{line-all:if-1.1}}{
      Print $\sus_k \leftarrow (start, length)$\label{line-all:if-1.2}\;}
    \lElseIf{$\sus_{k-1}.start + \sus_{k-1}.length -1 > k - 1$\label{line-all:if-2.1}}{
%      \tcp{The right boundary of $\sus_{k-1}$ is not
%        $S[k-1]$, so $\sus_k$ cannot be an LSUS  extension.}
      Print $\sus_k \leftarrow (start,length) \label{line-all:if-2.2}$;
      $flag \leftarrow 1$\;
   }
    \lElseIf{$(start,length)={\tt (null,null)}$\label{line-all:if-3.1}}{
      Print $\sus_k \leftarrow (\sus_{k-1}.start, \sus_{k-1}.length + 1)$\label{line-all:if-3.2}\;
    }
    \lElseIf{$length < \sus_{k-1}.length+1$\label{line-all:if-5.1}}{
      Print $\sus_k \leftarrow (start, length)$\label{line-all:if-5.2};
      $flag \leftarrow 1$\;
    }
    \ElseIf(\tcc*[f]{Print the LSUS, so it won't be lost due to the linked list 
            update at Line~\ref{line:del-start}--\ref{line:del-end} in Algorithm~\ref{algo:list}})
    {$length = \sus_{k-1}.length+1$}{
      Print $\sus_k \leftarrow (start, length)$;
      $flag \leftarrow 1$\;
    }
    \lElse
    {
      Print $\sus_k \leftarrow (\sus_{k-1}.start, \sus_{k-1}.length +
      1)$\label{line-all:if-4.2};
    }

\smallskip

    \tcc{Print out other SUSes that cover location $k$.}
    \If{$flag = 1$\label{line:all-1}}{
      \lIf{$\sus_{k-1}.length+1 = \sus_k.length$}{
        Print $(\sus_{k-1}.start, \sus_{k-1}.length + 1)$\;
      }

      $j\leftarrow head$\;
      \While{$j>0$ and $j\leq tail$}{
        \tcc{$List[j].start \neq \sus_k.start$ condition checking is because the SUS from {\tt head} node may have been printed.}
        \If{$List[j].length == \sus_k.length$ and $List[j].start \neq \sus_k.start$}{
          Print $(List[j].start,List[j].length)$;
          $j\leftarrow j+1$\;
        }
        \lElse{Break;\label{line:all-2}}
      }
    }
}
}%\small   
\end{algorithm}

\newpage

\subsubsection*{Proof for Lemma~\ref{lem:lsus}.}

\begin{proof}
  Note that $L_i$ is the length of the lcp between the suffix
  $S[i\ldots n]$ and any other suffix of $S$.  If $i+ L_i \leq n$, it
  means substring $S[i\ldots i+L_i]$ exists and is unique, while
  substring $S[i\ldots i+L_i-1]$ is either empty or is a repeat, so
  $S[i\ldots i+L_i]$ is $\lsus_i$. On the other hand, if $i+ L_i > n$,
  it means $S[i\ldots i+L_i-1]$ is indeed the suffix $S[i\ldots n]$
 and is a repeat, so $\lsus_i$ does not exist. \qed
\end{proof}

\subsubsection*{Proof for Lemma~\ref{lem:ext}}
\begin{proof}
  Let's say we are looking at $\sus_k$ for any $k\in \{1\ldots n\}$. 
We know $\sus_k$ exists for any $k$, so let's say $\sus_k = S[i\ldots
j]$, $1\leq i\leq k \leq j\leq n$. If $S[i\ldots j]$ is neither
$\lsus_i$ nor an extension of $\lsus_i$, it means $S[i\ldots j]$ is a
repeat. It is a contradiction because $S[i\ldots j] = \sus_k$, which
is unique.  \qed
\end{proof}

\subsubsection*{Proof for Lemma~\ref{lem:one}}
\begin{proof}
  The procedure starts with the candidate $S[1\ldots n]$, which is
  indeed unique (Line~\ref{line:start}). Then the
  {\tt For} loop calculates the $\lsus_i$ for $i=1,2,\ldots,k$
  (Lemma~\ref{lem:lsus}). If $\lsus_i$ exists (Line~\ref{line:if-1})
  and the length of $\lsus_i$ or its up-to-$k$ extension is less than
  the length of the current best candidate (Line~\ref{line:if-2}),
  then we will pick that $\lsus_i$ or its up-to-$k$ extension as the
  new candidate for $\sus_k$. This also resolves the possible ties by
  picking the leftmost candidate.  In the end of the procedure, we
  will have the shortest one among $\lsus_1\ldots \lsus_k$ or their
  up-to-$k$ extensions, and that is $\sus_k$. Early stop is made at
  Line~\ref{line:earlystop} if the $\lsus$ being calculated does not
  exist, because that means all the remaining $\lsus$es to be
  calculated do not exist either. Each step in the {\tt For} loop
  costs $O(1)$ time and the loop executes no more than $k$ steps, so
  the procedure takes $O(k)$ time. \qed
\end{proof}

\subsubsection{Proof for Lemma~\ref{lem:ext2}}
\begin{proof}
  Because $\sus_k$ is an extension of an LSUS, we have
  $\sus_k=S[i\ldots k]$ for some $i < k$ and $\lsus_i = S[i\ldots j]$
  for some $j<k$.  We also know $S[i\ldots k-1]$ is unique, because
  the unique substring $S[i\ldots j]$ is a prefix of $S[i\ldots k-1]$.
  Note that any substring starting from a location before $i$ and
  covering location $k-1$ is longer than the unique substring
  $S[i\ldots k-1]$, so $\sus_{k-1}$ must be starting from a location
  between $i$ and $k-1$, inclusive.
  Next, we show $\sus_{k-1}$ actually must start at location $i$.  The
  fact $\sus_k = S[i\ldots k]$ tells us that $|\lsus_t| \geq |\sus_k|
  = k-i+1$ for every $t=i+1, i+2, \ldots, k$; otherwise, any $\lsus_t$
  that is shorter than $k-i+1$ would be a better candidate than
  $S[i\ldots k]$ as $\sus_k$.  That means, any unique substring
  starting from $t=i+1, i+2, \ldots, k-1$ has a length at least
  $k-i+1$. However, $|S[i\ldots k-1]| = k-i < k-i+1$ and $S[i\ldots
  k-1]$ is unique already and covers location $k-1$ as well, so
  $S[i\ldots k-1]$ is the only candidate for $\sus_{k-1}$. This also
  means $\sus_k$ is indeed the substring $\sus_{k-1}$ appended by $S[k]$.
  \qed
\end{proof}

\subsubsection*{Proof for Lemma~\ref{lem:exist}}
\begin{proof}
(1) $\lsus_1$ must exist, because the string $S$ can be
$\lsus_1$ if every proper prefix of $S$ is a repeat. 
(2) If $\lsus_k$ exists, say $\lsus_k = S[k \ldots \gamma_k]$, then 
$\lsus_i$ exists for every $i\leq k$, because at least $S[i \ldots \gamma_k]$ is
unique due to the fact that, $S[k \ldots \gamma_k]$ is unique and also is a suffix
of $S[i \ldots \gamma_k]$. 
(3) If $\lsus_k$ does not exist, it means $S[k \ldots n]$ is a repeat,
and thus every suffix $S[i\ldots n]$ of $S[k \ldots n]$ for $i<k$ is
also a repeat, indicating $\lsus_i$ does not exist for every $i\geq k$.
\qed
\end{proof}

\subsubsection*{Proof for Lemma~\ref{lem:list-time}}
\begin{proof}
  All operations in $\mathit{FindSLS}(k)$ clearly take $O(1)$ time,
  except the {\tt while} loop at Line~\ref{line:while}, which is to
  merge linked list nodes whose candidates can be shorter. Thus, the
  lemma will be proved, if we can prove the amortized number of linked nodes
  that will be merged via that
  {\tt while} loop is also bounded by a constant. Note that any node
  in the linked list never splits. In the sequence of function
  calls
  $\mathit{FindSLS}(1), \ldots,
  \mathit{FindSLS}(n)$, there are at most $n$ linked list nodes to be
  merged. We know 
  the number of merge operations in merging $n$ nodes into one node
  (in the worst case) is
  no more than $O(n)$. So the amortized cost on merging nodes in one 
 $\mathit{FindSLS}()$ function call is $O(1)$. This finishes the proof
 of the lemma. 
\remove{
  The only scenario where the {\tt while} loop will take steps is: The
  linked list is not empty right before the ${\tt
    \mathit{FindSLS(k)}}$ function call and $\lsus_k$ exists.  Note
  that $|\lsus_k|$ will be the new node's candidate's length if
  merging happens. However, (1) $|\lsus_k|$ will be at least
  $|\lsus_{k-1}|-1$ (Lemma~\ref{lem:lsus2}), and (2) the tail node's
  candidate's length is at most $\lsus_{k-1}$ before the ${\tt
    \mathit{FindSLS(k)}}$ function call, and (3) The lengths of the
  candidates in different nodes are all distinct. The combination of
  the three conditions means the number of nodes that can be merged is
  no more than $3$, which also bounds the number of steps in the {\tt
    while} loop.
}
\qed
\end{proof}

\subsubsection*{Proof for Theorem~\ref{thm:one}}
\begin{proof}
  The suffix array of $S$ can be constructed by existing algorithms
  using $O(n)$ time and space (For ex., \cite{KA-SA2005}). After the
  suffix array is constructed, the rank array can be trivially created
  using $O(n)$ time and space.  We can then use the suffix array and
  the rank array to construct the lcp array using another $O(n)$ time
  and space~\cite{KLAAP01}.  Combining the time cost of
  Algo.~\ref{algo:one} (Lemma~\ref{lem:one}), the total time cost for
  finding $\sus_k$ for any location $k$ in the string $S$ is $O(n)$
  with a total of $O(n)$ space usage. If multiple candidates for
  $\sus_k$ exist, the leftmost candidate will be returned as is
  provided by Algo.~\ref{algo:one} (Lemma~\ref{lem:one}). \qed
\end{proof}

\subsubsection*{Proof for Theorem~\ref{thm:time}}
\begin{proof}
  We can construct the suffix array of the string $S$ in a total of
  $O(n)$ time and space using existing algorithms (For ex.,
  \cite{KA-SA2005}).  The rank array is just the inverse suffix array
  and can be directly obtained from SA using $O(n)$ time and space. Then
  we can obtain the lcp array from the suffix array and rank array
  using another $O(n)$ time and space~\cite{KLAAP01}. So the total
  time and space costs for preparing these auxiliary data structures
  are $O(n)$.

  \emph{Time cost.} The amortized time cost for each $\mathit{FindSLS}$ function call
  at Line~\ref{line:update}  in the sequence of function
  calls 
  $\mathit{FindSLS}(1),  \ldots,
  \mathit{FindSLS}(n)$  is $O(1)$
 (Lemma~\ref{lem:list-time}).
  The time cost for Line~\ref{line:if-1.1}--\ref{line:if-4.2} is also
  $O(1)$.  There are a total of $n$ steps in the {\tt For} loop,
  yielding a total of $O(n)$ time cost. 

  \emph{Space usage.}  The only space usage (in addition to  
the auxiliary data structures such as suffix array, rank
 array, and the lcp array, which cost a total of $O(n)$
 space)
in our algorithm is the
 dynamic  linked list, which however has no more than $n$ node at any time. Each
 node costs $O(1)$ space. Therefore, the linked list costs $O(n)$
 space. Adding the space usage of the auxiliary data structures, we
 get the total space usage of finding every SUS is $O(n)$.  \qed
\end{proof}

\end{document}